\documentclass[a4paper,11pt]{article}
\usepackage{amsmath,amssymb,amsfonts,enumerate}
\usepackage[thmmarks,amsmath]{ntheorem}
\usepackage{url}
\usepackage{forest}
\usepackage[center,small]{caption}
\usepackage{times}
\usepackage[ruled,lined]{algorithm2e}
\usepackage{authblk}

\newcommand{\mc}{\mathcal}
\newcommand{\alp}{{\mathsf{alph}}}
\newcommand{\rhs}{{\mathsf{rhs}}}
\newcommand{\pos}{{\mathsf{pos}}}
\newcommand{\reduce}{{\mathsf{reduce}}}
\newcommand{\val}{{\mathsf{val}}}
\newcommand{\dA}{{\mathcal{G}}}
\newcommand{\dB}{{\mathcal{H}}}
\newcommand{\labels}{\mathrm{labels}}
\newcommand{\nat}{\mathbb{N}}
\def\LZ78{${\mathsf{LZ78}}$}

\theoremstyle{plain}
\theorembodyfont{}
\newtheorem{proposition}{Proposition}
\newtheorem{theorem}{Theorem}
\newtheorem{example}{Example}
\newtheorem{lemma}{Lemma}
\theoremstyle{nonumberplain}
\theoremsymbol{\ensuremath{\square}}
\newtheorem{proof}{Proof}

\forestset{.style={for tree={s sep=5mm}}}

\begin{document}

\title{Traversing Grammar-Compressed Trees with\\ Constant Delay}

\author[1]{Markus Lohrey}
\author[2]{Sebastian Maneth}
\author[1]{Carl Philipp Reh}

\affil[1]{Universit\"at Siegen, Germany}
\affil[2]{University of Edinburgh, UK}

\date{}

\maketitle

\begin{abstract}
A grammar-compressed ranked tree is represented with a linear space overhead so that
a single traversal step, i.e., the move to the parent or the $i$th child, can be carried out in constant time.
Moreover, we extend our data structure such that equality of subtrees can be checked in constant time.
\end{abstract}

\section{Introduction}
Context-free grammars that produce single strings are a widely studied compact string representation
and are also known as  {\em straight-line programs} (SLPs).
For instance, the string $(ab)^{1024}$ can be represented
by the SLP with the rules $A_0 \to ab$ and $A_i \to A_{i-1} A_{i-1}$ for $1 \leq i \leq 10$
($A_{10}$ is the start
symbol).
In general, an SLP of size $n$ can produce a string of length $2^{\Omega(n)}$.
Besides grammar-based compressors
(e.g. \LZ78, $\mathsf{RePair}$, or $\mathsf{BISECTION}$, see \cite{CLLLPPSS05} for more details)
that derive an SLP from a given string, also algorithmic problems on SLP-compressed strings
such as pattern matching, indexing, and compressed word problems
have been investigated thoroughly, see \cite{Loh12survey} for a survey.

Motivated by applications where large tree structures occur, like XML processing,
SLPs have been extended to node-labelled ranked ordered trees~\cite{MLMN13,BuLoMa07,HuckeLN14,JezLo14approx,LohreyMM13}.
In those papers, straight-line linear context-free tree grammars are used.
Such grammars produce a single tree and are also known as tree straight-line programs~(TSLPs).
TSLPs generalize dags (directed acyclic graphs), which are widely used as compact tree representation. Whereas
dags only allow to share repeated subtrees, TSLPs can also share repeated internal tree patterns
(i.e., connected subgraphs).
The grammar-based tree compressor from~\cite{HuckeLN14} produces for every tree
(for a fixed set of node labels) a TSLP of size  $\mc O(\frac{n}{\log n})$ and height $O(\log n)$, which is worst-case optimal.
Various querying problems on TSLP-compressed trees such as
XPath querying and evaluating tree automata are studied in~\cite{LoMa06,LoMaSS12,ManethS13}.

In this paper we study the problem of navigating in a TSLP-represented tree:
Given a TSLP $\dA$ for a tree $t$,
the task is to precompute in  time $O(|\dA|)$ an $O(|\dA|)$-space data structure
that allows to move from a node of $t$ in time $O(1)$ to its parent node or to its $i$th child
and to return in time $O(1)$ the node label of the current node.
Here the nodes of $t$ are represented in space $O(|\dA|)$ in a suitable way.
Such a data structure has been developed for string SLPs in~\cite{GasieniecKPS05};
it allows to move from left to right over the string produced by the SLP
requiring time $O(1)$ per move.
We first extend the data structure from~\cite{GasieniecKPS05} so that
the string can be traversed in a two-way fashion, i.e.,
in each step we can move either to the left or right neighboring position in constant time.
This data structure is then used to navigate in a TSLP-represented tree.

TSLPs are usually used for the compression of ranked trees, i.e., trees where the maximal number $r$ of children of a node
is bounded by a constant. For instance, the above mentioned bound
$\mc O(\frac{n}{\log n})$ from \cite{HuckeLN14} assumes a constant $r$.
In many applications, $r$ is indeed bounded (e.g., $r=2$ for the following two encodings of unranked trees).
For unranked trees where $r$ is unbounded, it is more realistic to require that the data structure supports
navigation to
(\emph{i})~the parent node,
(\emph{ii})~the first child, (\emph{iii})~the right sibling, and (\emph{iv})~the left sibling.
We can realize these operations by using a suitable constant-rank encoding of trees.
Two folklore binary tree encodings of an unranked tree $t$ with maximal rank $r$ are:
\begin{itemize}
\item
First-child/next-sibling encoding $\mathsf{fcns}(t)$:
The left (resp. right) child of a node in $\mathsf{fcns}(t)$ is the first child (resp., right sibling) in $t$.
On this encoding, we can support $O(1)$ time navigation for (\emph{ii})--(\emph{iv}) of above.
The parent move (\emph{i}) however, requires $O(r)$ time.
\item
Binary encoding:
We define the binary encoding $\mathrm{bin}(t)$ by adding for every node $v$ of rank $s \leq r$
a binary tree of depth $\lceil \log s \rceil$ with $s$ many leaves,
whose root is $v$ and whose leaves are the children of $v$.
This introduces at most $2s$ many new binary nodes (labelled by a new symbol).
Thus $|\mathrm{bin}(t)| \leq 3 |t|$.
Every navigation step in the original tree can be simulated by $O(\log r)$ many navigation
steps in $\mathrm{bin}(t)$.
\end{itemize}
Our second main result concerns subtree equality checks. This is the problem of checking for two
given nodes of a tree, whether the subtrees rooted at these two node are identical. We extend our
data structure for tree navigation such that subtree equality checks can be done in time $O(1)$.
The problem of checking equality of subtrees occurs in several different contexts, see for instance
\cite{CaiP95} for details. Typical applications are common subexpression detection, unification, and
non-linear pattern matching. For instance,
checking whether the pattern $f(x,f(y,y))$ is matched at a certain tree node needs a constant number of navigation
steps and a single subtree equality check.

\medskip
\noindent
{\bf Further related work.} \quad
The ability to navigate efficiently in a tree is a basic prerequisite
for most tree querying procedures. For instance, the DOM representation
available in web browsers through JavaScript,
provides tree navigation primitives (see, e.g.,~\cite{DBLP:conf/edbt/DelprattRR08}).
Tree navigation has been intensively studied in the context of succinct tree representations.
Here, the goal is to represent a tree by a bit string, whose length is asymptotically equal to the
information theoretic lower bound. For instance, for binary trees of $n$ nodes,
the information theoretic lower bound is $2n + o(n)$ and there exist
succinct representations that encode a binary tree of size $n$ by a bit string of
length $2n + o(n)$.  In addition there exist such encodings that allow to navigate in the tree in constant time
(and support many other tree operations), see e.g.~\cite{DBLP:journals/talg/NavarroS14} for a survey.

Recently, grammatical formalisms for the compression of unranked trees have been proposed as well.
In~\cite{BilleGLW13} the authors consider so called top dags as a compact tree representation.
Top dag can be seen as a slight variant of TSLPs for an unranked tree.
It is shown in~\cite{BilleGLW13} that for every tree of size $n$ the top dag has size
$O(\frac{n}{\log^{0.19} n})$. Recently, this bound was improved to $O\big(\frac{n \cdot \log \log n}{\log n}\big)$
in \cite{HSR15} (it is open whether this bound can be improved to the information theoretic limit
$O\big(\frac{n}{\log n}\big)$). Moreover, also the navigation problem for
top dags is studied in~\cite{BilleGLW13}. The authors show that
a single navigation step in $t$ can be done in time $O(\log |t|)$ in the top dag. Nodes
are represented by their preorder numbers, which need $O(\log |t|)$ bits.
In \cite{BLRSSW15} an analogous result has been shown for unranked trees that are represented
by a string SLP for the balanced bracket representation of the tree. This covers also TSLPs: From a
TSLP $\dA$ one can easily compute in linear time an SLP for the balanced bracket representation
of $\val(\dA)$. In some sense our results are orthogonal to the results of  \cite{BLRSSW15}.
\begin{itemize}
\item We can navigate, determine node labels, and check equality of subtrees in time $O(1)$,
but our representation of tree nodes needs space $O(|\dA|)$.
\item Bille et al.~\cite{BLRSSW15} can navigate and do several other tree queries (e.g. lowest common ancestor 
computations) in time $O(\log |t|)$,
but their node representations (preorder numbers) only need space $O(\log |t|) \leq O(|\dA|)$.
\end{itemize}
An implementation of navigation over TSLP-compressed trees is given in~\cite{MaSe10}. Their worst-case time
per navigation step is $O(h)$ where $h$ is the height of the TSLP. The authors demonstrate
that on XML data full traversals are $5$--$7$ times slower than over succinct trees (based on an
implementation by Sadakane) while using $3$--$15$ times less space.

Checking equality of subtrees is trivial for minimal dags, since every subtree is uniquely represented.
For so called SL grammar-compressed dags (which can be seen as TSLPs with certain restrictions)
it was shown in \cite{MLMN13} that equality of subtrees can be checked in time $O(\log |t|)$ for given preorder numbers.

\section{Preliminaries}

For an alphabet $\Sigma$ we denote by $\Sigma^*$ the set of all
strings over $\Sigma$ including the empty string $\epsilon$.
For a string $w=a_1\cdots a_n$ ($a_i \in \Sigma$) we denote by $\alp(w)$ the
set of symbols $\{a_1,\dots,a_n\}$ occurring in $w$.
Moreover, let $|w|=n$, $w[i] = a_i$ and $w[i : j] = a_i \cdots a_j$
where $w[i : j] = \varepsilon$, if $i > j$. Let $w[:i] = w[1:i]$ and $w[i:] = w[i:n]$.

\section{Straight-Line Programs}
A \emph{straight-line program (SLP)}, is a triple
$\dB = (N,\Sigma, \rhs)$, where $N$ is a
finite set of \emph{nonterminals}, $\Sigma$ is a finite set of \emph{terminals} ($\Sigma \cap N = \emptyset$),
and $\rhs : N \to (N \cup \Sigma)^*$ is a mapping such that the binary relation
$\{ (A, B) \in N \times N \mid B \in \alp(\rhs(A)) \}$ is acyclic.
This condition ensures that every nonterminal $X \in N$ produces a unique string $\val_{\dB}(X) \in \Sigma^*$.
It is obtained from the string $X$ by repeatedly replacing nonterminals $A$ by $\rhs(A)$,
until no nonterminal occurs in the string. We will also write $A \to \alpha$ if $\rhs(A) = \alpha$ and call it a {\em rule} of $\dB$.
Usually, an SLP has a start nonterminal as well, but for our purpose, it is more convenient to
consider SLPs without a start nonterminal.

The size of the SLP $\dB$ is $|\dB| = \sum_{A \in N} |\rhs(A)|$, i.e., the total length of all right-hand sides.
A simple induction shows that for every SLP $\dB$ of size $m$ and every nonterminal $A$,
$|\val_{\dB}(A)|\in O(3^{m/3})$ \cite[proof of Lemma~1]{CLLLPPSS05}.
On the other hand, it is straightforward to define an SLP
$\dB$ of size $2n$ such that $|\val(\dB)| \geq 2^n$.
Hence, an SLP can be seen as a compressed representation of the string it generates,
and it can achieve exponential compression rates.

In Section~\ref{sec-equality} we will use some algorithmic facts about SLPs, which are collected in the
following proposition, see \cite{Loh12survey} for details:

\begin{proposition} \label{prop-SLP-alg}
There are polynomial time algorithms for the following problems:
\begin{itemize}
\item Given an SLP $\dB$ and a position $i$, compute the symbol $\val(\dB)[i]$.
\item Given an SLP $\dB$ and two positions $i < j$, compute an SLP for the string $\val(\dB)[i:j]$.
\item Given two SLPs $\dB_1$ and $\dB_2$, compute the length of the longest common prefix of $\val(\dB_1)$ and
$\val(\dB_2)$.
\end{itemize}
\end{proposition}

\section{Tree Straight-Line Programs}
For every $i \ge 0$, we fix a countable infinite set $\mc F_i$ of \emph{terminals} of rank $i$
and a countable infinite set $\mc N_i$ of \emph{nonterminals} of rank $i$. Let
$\mc F=\bigcup_{i\ge 0} \mc F_i$ and $\mc N=\bigcup_{i\ge 0} \mc N_i$.
Moreover, let $\mc X=\left\{x_1,x_2,\dots\right\}$ be the set of \emph{parameters}.
We assume that the three sets $\mc F$, $\mc N$, and $\mc X$ are pairwise disjoint.
A \emph{labeled tree} $t = (V,E,\lambda)$ is a finite, directed and ordered tree $t$
with set of nodes $V$, set of edges $E\subseteq V\times\nat\times V$,
and labeling function $\lambda:V\to\mc F\cup\mc N\cup \mc X$.
We require for every node $v\in V$ that if $\lambda(v) \in \mc F_k\cup \mc N_k$,
then $v$ has $k$ distinct children $u_1,\dots,u_k$, i.e.,
$(v,i,u)\in E$ if and only if $1 \leq i \leq k$ and $u=u_i$.
A leaf of $t$ is a node with zero children.
We require that every node $v$ with $\lambda(v) \in \mc X$ is a leaf of $t$.
The \emph{size} of $t$ is $|t|=|V|$.
We denote trees in their usual term notation, e.g. $a(b,c)$
denotes the tree with an $a$-labeled root node that has a first child
labeled $b$ and a second child labeled $c$.
We define $\mc T$ as the set of all labeled trees.
Let $\labels(t) = \{ \lambda(v) \mid  v \in V\}$. For $\mc L \subseteq \mc F\cup\mc N\cup \mc X$ we let
 $\mc T(\mc L) = \{ t \in \mc T \mid \labels(t) \subseteq \mc L\}$.
 The tree $t\in \mc T$ is \emph{linear} if there do not exist different leaves that are labeled with the same parameter.
We now define a particular form of context-free tree grammars (see \cite{tata97} for more
details on context-free tree grammars) with the property that exactly one tree is derived.
A \emph{tree straight-line program (TSLP)} is a triple $\mc G=(N, S, \rhs)$, where
$N \subseteq \mc N$ is a finite set of nonterminals,
$S \in \mc N_0 \cap N$
is the start nonterminal, and $\rhs : N \to \mc T(\mc F \cup N \cup \mc X)$ is a
mapping such that the following hold:
\begin{itemize}
\item For every $A \in N$, the tree $\rhs(A)$ is linear and if $A \in \mc N_k$ ($k \geq 0$) then
$\mc X \cap \labels(\rhs(A)) = \{ x_1, \ldots, x_k\}$.
 \item The binary relation $\{ (A,B) \in N \times N \mid  B \in  \labels(\rhs(A))\}$ is acyclic.
\end{itemize}
These conditions ensure that from every nonterminal $A \in N \cap \mc N_k$ exactly one linear tree
$\val_{\mc G}(A) \in \mc T(\mc F\cup\{x_1, \ldots, x_k\})$
is derived by using the rules $A \to \rhs(A)$ as rewrite rules in the usual sense.
More generally, for every tree $t \in \mc T(\mc F\cup\mc N \cup\{x_1, \ldots, x_n\})$
we can derive a unique tree $\val_{\mc G}(t) \in \mc T(\mc F\cup\{x_1, \ldots, x_n\})$
with the rules of $\mc G$.
The tree defined by $\mc G$ is
$\val(\mc G) = \val_{\mc G}(S)$.
Instead of giving a formal definition, we show a derivation of $\val(\mc G)$ from $S$ in an example:
\begin{example}\label{example:SLCF}
Let $\mc G = (\{ S,A,B,C,D,E,F\}, S, \rhs)$, $a\in\mc F_0$, and $b\in\mc F_2$, where
\begin{align*}
& \rhs(S) =  A(B),\;  \rhs(A) = C(F,x_1),\; \rhs(B) = E(F),\; \rhs(C) =  D(E(x_1),x_2),\\
& \rhs(D) = b(x_1,x_2),\; \rhs(E) = D(F,x_1),\; \rhs(F) = a.
\end{align*}
A possible derivation of $\val(\mc G) = b(b(a,a),b(a,a))$ from $S$ is:
\begin{align*}
S&\to A(B)\to C(F,B)\to D(E(F),B)\to b(E(F),B)\to b(D(F,F),B)\\
&\to b(b(F,F),B) \to
 b(b(a,F),B) \to b(b(a,a),B) \to b(b(a,a),E(F))\\
&\to b(b(a,a),D(F,F)) \to b(b(a,a),b(F,F))\to b(b(a,a),b(a,F))\\
& \to b(b(a,a),b(a,a)).
\end{align*}
\end{example}
As for SLPs, we will also write $A \to t$ if $\rhs(A) = t$.
The size $|\mc G|$ of a TSLP $\mc G=(N, S, \rhs)$ is defined as
$|\mc G|=\sum_{A\in N}|\rhs(A)|$ (recall that we do not count parameters).
For instance, the TSLP from Example~\ref{example:SLCF} has size $12$.

A TSLP $\mc G = (N, \rhs, S)$ is {\em monadic} if $N \subseteq \mc N_0 \cup \mc N_1$, i.e., every
nonterminal has rank at most 1.
The following result is shown in \cite{LoMaSS12}.

\begin{proposition} \label{thm-monadic}
From a given TSLP $\mc G$, where $r$ and $k$ are the maximal ranks of terminal and nonterminal symbols
appearing in a right-hand side, one can construct in
time $O(r \cdot k \cdot |\mc G|)$ a monadic TSLP $\mc H$ such that
$\val(\mc H) = \val(\mc G)$ and $|\mc H|\in O(r \cdot |\mc G|)$.
\end{proposition}
Finally, we need the following algorithmic result from \cite{BuLoMa07}:

\begin{proposition} \label{prop-TSLP-eq}
For two TSLPs $\dA_1$ and $\dA_2$ we can check in polynomial time whether $\val(\dA_1) = \val(\dA_2)$.
\end{proposition}

\section{Computational model} \label{sec-comp-model}

We use the word RAM model in the following sections, where registers have a certain bit length $w$.
Arithmetic operations and comparisons of registers can be done in time $O(1)$.
The space of a data structure is measured by the number of registers. Our algorithms
need the following register lengths $w$, where $\dA$ is the input TSLP and $t = \val(\dA)$.
\begin{itemize}
\item For navigation (Section~\ref{sec-TSLP-trav}) we need a bit length of $w = O(\log |\dA|)$,
since we only have to store numbers of length at most $|\dA|$.
\item For equality checks  (Section~\ref{sec-equality})  we need a bit length of $w = O(\log |t|) \leq O(|\dA|)$, which is the same
assumption as in \cite{BilleGLW13,BLRSSW15}.
\end{itemize}

\section{Two-Way Traversal in SLP-Compressed Strings} \label{sec-string-trav}

In~\cite{GasieniecKPS05} the authors present a data structure of size $O(|\dB|)$ for storing an SLP~$\dB$
that allows to produce $\val_{\dB}(A)$ with time delay of $O(1)$ per symbol. That is, the symbols of
$\val_{\dB}(A)$ are produced from left to right and for each symbol constant time is needed.

In a first step, we enhance the data structure from \cite{GasieniecKPS05} for traversing SLP-compressed strings
in such a way that both operations of moving to the left and right symbol are supported in constant time.
For this, we assume that the SLP $\dB = (N,\Sigma,\rhs)$ has the property that
$|\rhs(A)|=2$ for each $A\in N$.
Every SLP $\dB$ with $|\val(\dB)| \geq 2$
can easily be transformed in linear time into an SLP $\dB'$ with this property and
$\val(\dB')=\val(\dB)$:
First, we replace all occurrences of nonterminals $B$ with $|\rhs(B)| \leq 2$ by $\rhs(B)$.
If $\rhs(S) = A \in N$, then we redefine $\rhs(S) = \rhs(A)$.
Finally, for every $A \in N$ such that $\rhs(A)=\alpha_1\cdots \alpha_n$ with $n\geq 3$
and $\alpha_1,\dots,\alpha_n\in (N\cup\Sigma)$ we introduce new nonterminals $A_2,\dots, A_{n-1}$
with rules $A_i\to \alpha_iA_{i+1}$ for $2\leq i\leq n-2$, and $A_{n-1}\to \alpha_{n-1}\alpha_n$, and
redefine $\rhs(A)=\alpha_1A_2$.
It should be clear that $|\dB'| \leq 2 \cdot |\dB|$.

Note that the positions in $\val_{\dB}(X)$
correspond to root-leaf paths in the (binary) derivation tree of $\dB$ that is rooted in the nonterminal $X$.
We represent such a path by merging successive edges where the path moves in the same direction
(left or right) towards the leaf.
To formalize this idea, we define for every $\alpha\in N \cup \Sigma$ the strings $L(\alpha), R(\alpha) \in N^*\Sigma$
inductively as follows:  For $a \in \Sigma$ let
\begin{gather*}
L(a)  =  R(a) = a.
\end{gather*}
For $A \in N$ with $\rhs(A) = \alpha\beta$ ($\alpha, \beta \in N \cup \Sigma$) let
\begin{equation} \label{L-and-R}
L(A)  =  A \, L(\alpha) \text{ and }
R(A)  = A \, R(\beta) .
\end{equation}
Note that for every $A \in N$, the string $L(A)$ has the form
$A_1 A_2 \cdots A_n a$ with $A_i \in N$, $A_1 = A$, and $a \in \Sigma$.
We define $\omega_L(A) = a$. \label{page-omega-L} The terminal $\omega_R(A)$ is defined
analogously by referring to the string $R(A)$.

\begin{figure}[t]
\begin{center}
\begin{forest}
for tree={edge=->},
phantom[
  [$a$ [$C$ [$B$ [$A$[ $S$]]]] [$D$]]
  [$b$]
]
\end{forest}
\qquad
\begin{forest}
for tree={edge=->},
phantom[
  [$a$]
  [$b$ [$D$ [$C$ [$A$] [$D$ [$S$]]]]]
]
\end{forest}
\end{center}
\caption{\label{fig-ex-string} The tries $T_L(a)$, $T_L(b)$ (left), and $T_R(a)$, $T_R(b)$ (right) for the SLP from Example~\ref{ex-traversal-string}.}
\end{figure}
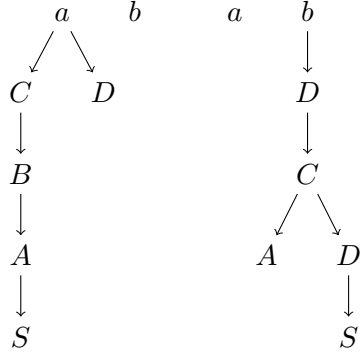

\begin{example} \label{ex-traversal-string}
Let $\mc H = (\{S,A,B,C,D\}, \{a,b\}, \rhs)$, where
$$
S \to AB, \  A \to BC, \ B \to CC, \ C \to aD, \ D \to ab .
$$
Then we have $L(S) = SABCa$, $L(A) = ABCa$, $L(B)=BCa$, $L(C) = Ca$, $L(D)=Da$,
$R(S)=SBCDb$, $R(A)=ACDb$, $R(B)=BCDb$, $R(C)=CDb$, $R(D)=Db$.
Moreover, $\omega_L(X) = a$ and $\omega_R(X) = b$ for all $X \in \{S,A,B,C,D\}$.
\end{example}
We store all strings $L(A)$ (for $A \in N$) in $|\Sigma|$ many
tries: Fix $a \in \Sigma$ and let $w_1, \ldots, w_n$ be all strings
$L(A)$ such that $\omega_L(A) = a$. Let $v_i$ be the string $w_i$ reversed.
Then, $a, v_1, \ldots, v_n$ is a prefix-closed set of strings (except that the empty
string is missing) that can be stored
in a trie $T_L(a)$. Formally, the nodes of $T_L(a)$ are the strings $a, v_1, \ldots, v_n$, where each node is
labeled by its last symbol (so the root is labeled with $a$),
and  there is an edge from $a w$ to $a w A$ for all appropriate $w \in N^*$ and $A \in N$.
The tries $T_R(a)$ are defined in the same way by referring to the strings $R(A)$.
Note that the total number of nodes in all tries $T_L(a)$ ($a \in \Sigma$) is exactly $|N|+|\Sigma|$.
In fact, every $\alpha \in N \cup \Sigma$ occurs exactly once as a node label in the
forest $\{ T_L(a) \mid a \in \Sigma\}$.

\begin{example}[Example~\ref{ex-traversal-string} continued]
The tries $T_L(a)$, $T_L(b)$, $T_R(a)$, and $T_R(b)$ for the SLP from Example~\ref{ex-traversal-string} are shown
in Figure~\ref{fig-ex-string}.
\end{example}
Next, we define two alphabets $L$ and $R$ as follows:
\begin{align}
L =& \{ (A,\ell,\alpha) \mid \alpha \in \alp(L(A)) \setminus \{A\} \} \label{alph-L}  \\
R =& \{ (A,r,\beta) \mid \beta \in \alp(R(A)) \setminus \{A\} \} \label{alph-R}
\end{align}
Note that the sizes of these alphabets are quadratic in the size
of $\mc H$. On the alphabets $L$ and $R$ we define the
partial operations $\reduce_L : L \to L$ and $\reduce_R : R \to R$ as follows:
Let $(A,\ell,\alpha) \in L$. Hence we can write $L(A)$ as $A u \alpha v$ for some strings $u$ and $v$.
If $u = \varepsilon$, then $\reduce_L(A,\ell,\alpha)$ is undefined. Otherwise,
we can write $u$ as $u' B$ for some $B \in N$.
Then we define $\reduce_L(A,\ell,\alpha) = (A,\ell,B)$.
The definition of $\reduce_R$ is analogous: If  $(A,r,\alpha) \in R$,
then we can write $R(A)$ as $A u \alpha v$ for some strings $u$ and $v$.
If $u = \varepsilon$, then $\reduce_R(A,r,\alpha)$ is undefined. Otherwise,
we can write $u$ as $u' B$ for some $B \in N$ and define
$\reduce_R(A,r,\alpha) = (A,r,B)$.

\begin{example}[Example~\ref{ex-traversal-string} continued]
The sets $L$ and $R$ are:
\begin{align*}
L = \{ & (S,\ell,A), (S, \ell,B), (S,\ell,C), (S,\ell,a), (A,\ell,B), \\
& (A,\ell,C), (A,\ell,a), (B,\ell,C), (B,\ell,a), (C,\ell,a), (D,\ell,a) \} \\
R = \{ & (S,r,B), (S,r,C), (S,r,D), (S,r,b), (A,r,C), (A,r,D), \\
& (A,r,b), (B,r,C), (B,r,D), (B,r,b), (C,r,D), (C,r,b),  (D,r,b) \}.
\end{align*}
For instance, $\reduce_L(S,\ell,a) = (S,\ell,C)$  and $\reduce_R(B,r,D) = (B,r,C)$
whereas $\reduce_L(S,\ell,A)$ is undefined.
\end{example}
An element $(A,\ell,\alpha)$ can be represented by a pair $(v_1, v_2)$ of different nodes in the forest
$\{ T_L(a) \mid a \in \Sigma \}$, where $v_1$ (resp. $v_2$) is the unique node labeled with
$\alpha$ (resp., $A$). Note that $v_1$ and $v_2$ belong to the same trie and that $v_2$ is below
$v_1$. This observation allows us to reduce the computation of the mapping $\reduce_L$ to
a so-called \emph{next link query}:
From the pair $(v_1, v_2)$ we have to compute the unique child $v$ of $v_1$
such that $v$ is on the path from $v_1$ to $v_2$. If $v$ is labeled with $B$, then
$\reduce_L(A,\ell,\alpha) = (A,\ell,B)$, which is represented by the pair $(v,v_2)$.
Clearly, the same remark applies to the map $\reduce_R$.
The following result is mentioned in~\cite{GasieniecKPS05}, see Section~\ref{conc} for
a discussion.

\begin{proposition}\label{prop:gas}
A trie $T$ can be represented in space $O(|T|)$ such that any next link query can be answered in time
$O(1)$. Moreover, this representation can be computed in time $O(|T|)$ from $T$.
\end{proposition}
We represent a path in the derivation tree of $\dB$ with root $X$ by a sequence of triples
$$
\gamma = (A_1,d_1, A_2) (A_2,d_2, A_3) \cdots (A_{n-1}, d_{n-1},A_n) (A_n,d_n,a) \in (L \cup R)^+
$$
such that $n \geq 1$ and the following properties hold:
\begin{itemize}
\item  $A_1 = X$, $a \in \Sigma$
\item  $d_i = \ell$ if and only if $d_{i+1} =r$ for all $1 \leq i \leq n-1$.
\end{itemize}
We call such a sequence a {\em valid $X$-sequence for $\dB$} in the following,
or briefly a valid sequence if $X$ is not important and $\dB$ is clear from the context.
Note that a valid $X$-sequence $\gamma$
indeed defines a unique path in the derivation tree rooted at $X$ that ends in a leaf
that is labelled with the terminal symbol $a$ if $\gamma$ ends
with $(A,d,a)$. This path, in turn, defines a unique position in the string $\val_{\dB}(X)$
that we denote by $\pos(\gamma)$.

\begin{algorithm}[t]
\SetKwComment{Comment}{(}{)}
\KwData{valid sequence $\gamma$}
$(A,d,a) := \text{pop}(\gamma)$ \;
\eIf{$d = \ell$}
  {expand-right$(\gamma, A,a)$}
  {\eIf{$\gamma = \varepsilon$}
     {\Return undefined}
     {$(A',\ell,A) := \text{pop}(\gamma)$ \;
      $\gamma :=$ expand-right$(\gamma, A',A)$}}
  \Return$\gamma$
\caption{right$(\gamma)$  \label{right}}
\end{algorithm}

\begin{algorithm}[t]
\SetKwComment{Comment}{(}{)}
\KwData{sequence $\gamma$, $A \in N$, $\alpha \in N \cup \Sigma$ such that
$\gamma (A,\ell, \alpha)$ is a prefix of a valid sequence}
let $\rhs(A) = \alpha_1 \alpha_2$\;
\eIf{$\alpha \neq \alpha_1$}
  {$(A,\ell,B) := \reduce_L(A,\ell,\alpha)$ \;
   push$(\gamma, (A,\ell,B))$ \;
   let $\rhs(B) = \beta_1 \beta_2$\;
   push$(\gamma, (B,r,\beta_2))$\;
   \If{$\beta_2 \in N$}{push$(\gamma,  (\beta_2,\ell,\omega_L(\beta_2)))$}}
  {\eIf{$\gamma = \varepsilon$}
     {push$(\gamma, (A,r,\alpha_2))$}
     {$(B,r,A) := \text{pop}(\gamma)$ \;
      push$(\gamma, (B,r,\alpha_2))$}
       \If{$\alpha_2\in N$}{push$(\gamma, (\alpha_2,\ell,\omega_L(\alpha_2)))$}}
    \Return$\gamma$
\caption{expand-right$(\gamma, A,\alpha)$  \label{exp-right}}
\end{algorithm}

We now define a procedure $\mathsf{right}$ (see Algorithm~\ref{right}) that
takes as input a valid $X$-sequence $\gamma$ and returns a valid $X$-sequence $\gamma'$
such that $\pos(\gamma') = \pos(\gamma)+1$ in case the latter
is defined (and otherwise returns ``undefined'').
The idea uses the obvious fact that in order to move in a binary tree from a leaf to the next leaf
(where ``next'' refers to the natural left-to-right order on the leaves)
one has to repeatedly move to parent nodes as long as right-child edges are traversed
(in the opposite direction);
when this is no longer possible, the current node is the left child of its parent $p$.
One now moves to the right child of $p$ and from here repeatedly to left-children until a leaf is reached.
Each of these four operations can be implemented in constant time on valid sequences,
using the fact that consecutive edges to left (resp., right)
children are merged into a single triple from $L$ (resp., $R$) in our representation of paths.
We use the valid sequence $\gamma$ as a stack with the operations pop (which returns the right-most
triple of $\gamma$, which is thereby removed from $\gamma$) and push (which appends a given triple
on the right end of $\gamma$).
The procedure $\mathsf{right}$ uses the procedure expand-right (see Algorithm~\ref{exp-right})
that takes as input a (non-valid) sequence $\gamma$ of triples,
$A \in N$, and a symbol $\alpha \in N \cup \Sigma$
such that $\gamma (A,\ell, \alpha)$ is a prefix of a valid sequence. The sequence $\gamma$ has to be treated
as a global variable in order to obtain
an $O(1)$-time implementation (to make the presentation clearer, we pass $\gamma$
as a parameter to expand-right).

In a completely analogous way we can define a procedure $\mathsf{left}$ that takes as input a valid $X$-sequence $\gamma$
and returns a valid $X$-sequence $\gamma'$ such that $\pos(\gamma') = \pos(\gamma)-1$ in case the latter
is defined (otherwise ``undefined'' will be returned). The details are left to the reader.

\section{Traversal in  TSLP-Compressed Trees} \label{sec-TSLP-trav}

In this section, we extend the traversal algorithm from the previous section from SLPs to TSLPs.
We only consider monadic TSLPs. If the TSLP is not monadic, then we can transform it into a monadic
TSLP using Proposition~\ref{thm-monadic}.

Let us fix a {\em monadic} TSLP $\dA = (N, \rhs, S)$.
One can easily modify $\dA$ so that for all $A \in N$, $\rhs(A)$ has one of the following
four forms (we write $x$ for the parameter $x_1$):
\begin{enumerate}[(a)]
\item $B(C)$ for $B, C \in N$ (and $A$ has rank $0$)
\item $B(C(x))$ for $B, C \in N$ (and $A$ has rank $1$)
\item $a \in \mc F_0$ (and $A$ has rank $0$) 
\item $f(A_1, \ldots, A_{i-1},x,A_{i+1},\ldots, A_n)$ for $A_1, \ldots, A_{i-1}, A_{i+1}, \ldots, A_n \in N$,
$f \in \mc F_n$, $n \geq 1$
 (and $A$ has rank $1$)
\end{enumerate}
This modification is done in a similar way as the transformation of SLPs at the
beginning of Section~\ref{sec-string-trav}:
Remove nonterminals $A \in \mc N_1$ with $\rhs(A)=x_1$ (resp., $\rhs(A) = B(x_1)$, $B \in \mc N$)
by replacing in all right-hand sides
every subtree $A(t)$ by $t$ (resp., $B(t)$) and iterate this as long
as possible. Similarly, we can eliminate nonterminals $A \in \mc N_0 \setminus \{S\}$ with $\rhs(A) \in \mc N$.
If $\rhs(S) = B \in \mc N_0$ then redefine $\rhs(S) := \rhs(B)$. Finally, right-hand sides
of size at least two are split top-down into new nonterminals with right-hand sides of the above types.
In case the right-hand side contains the parameter $x_1$ (in a unique position),
we decompose along the path to $x_1$.  For example, the rule $Z(x) \to h(f(A,a),f(A,g(x)),B(A))$
is replaced by $Z(x) \to C(D(x))$, $C(x) \to h(E,x,F)$, $D(x) \to G(H(x))$, $E\to G(J)$,
$F \to B(A)$, $G(x) \to f(A,x)$, $H(x) \to g(x)$, $J \to a$.
The resulting TSLP has size at most $2|\dA|$.

Let us write $N_{x}$ ($x \in \{a,b,c,d\}$) for the set of all nonterminals
whose right-hand side is of the above type (x).
Let $N_1 = N_a \cup N_b$ and  $N_2  = N_c \cup N_d$.
Note that if we start with a nonterminal $A \in N_a$ and then
replace nonterminals from $N_1$ by their right-hand sides repeatedly, we obtain a tree that consists
of nonterminals from $N_d$ followed by
a single nonterminal from $N_c$.
After replacing these nonterminals by their right-hand sides,
we obtain a caterpillar tree which is composed of right-hand sides
of the form (d) followed by a single constant from $\mc F_0$. Hence, there
is a unique path of terminal symbols
from $\mathcal F$, and we call this path the {\em spine path} of $A$.
All other nodes of the caterpillar tree are leaves and labeled with nonterminals of rank zero to which we can apply
again the TSLP rules.
The size of a caterpillar tree and therefore a spine path can be exponential in the size of the TSLP.

We define an SLP $\dB = (N_1, N_2, \rhs_1)$ as follows:
If $A \in N_1$ with $\rhs(A) = B(C)$ or  $\rhs(A) = B(C(x))$, then $\rhs_1(A) = BC$.
The triple alphabets $L$ and $R$ from \eqref{alph-L} and \eqref{alph-R}  refer to this SLP $\dB$.
Moreover, we define $M$ to be the set of all triples $(A,k,A_k)$ where $A \in N_d$, $\rhs(A) = f(A_1, \ldots, A_{i-1},x,A_{i+1},\ldots, A_n)$
and $k \in \{1,\ldots,n\} \setminus \{i\}$.

Note that the nodes of the tree $\val(\dA)$ can be identified with the nodes of $\dA$'s
derivation tree that are labeled with a nonterminal from $N_2$ (every nonterminal from $N_2$
has a unique occurrence of a terminal symbol on its right-hand side).

A {\em valid sequence} for $\dA$ is a sequence
$$
\gamma = (A_1,e_1, A_2) (A_2,e_2, A_3) \cdots (A_{n-1}, e_{n-1},A_n) (A_n,e_n,A_{n+1}) \in (L \cup R \cup M)^*
$$
(note that $e_1, \ldots, e_n \in \{\ell,r\} \uplus \mathbb{N}$) such that $n \geq 0$ and
the following hold:
\begin{itemize}
\item If $S \in N_a$ then $n \geq 1$.
\item If $n \geq 1$ then $A_1 = S$ and $A_{n+1} \in N_2$.
\item If $e_i, e_{i+1} \in \{\ell,r\}$ then $e_i = \ell$ if and only if $e_{i+1} = r$.
\end{itemize}
Such a valid sequence represents a path in the derivation of the TSLP $\dA$ from the root to an $N_2$-labelled  node, and
hence represents a node of the tree $\val(\dA)$. Note that in case $S \in N_c$ the empty sequence is valid
too and represents the unique node of the single-node tree $\val(\dA)$. Moreover, if the last triple $(A_n,e_n,A_{n+1})$ belongs to $M$,
then we must have $A_{n+1} \in N_c$, i.e., $\rhs(A_{n+1}) \in \mc F_0$.
Here is an example:
\begin{example} \label{ex-tree-traversal}
Consider the monadic TSLP $\dA$ with nonterminals $S$, $A$, $B$, $C$, $D$, $E$, $F$, $G$, $H$, $J$ and the following
rules
\begin{align*}
& S \to A(B),\
A \to C(D(x)),\
B \to C(E),\
C \to f(F,x),\
D \to f(x,F),  \\
& E \to D(F),\
F \to G(H),\
G \to J(J(x)),\
H \to a,\
J \to g(x).
\end{align*}
It produces the tree shown in Figure~\ref{fig-tree-ex}.
We have $N_1 = \{ S,A,B,E,F,G\}$ and $N_2 = \{ C,D,J\}$.
The SLP $\dB$ consists of the rules
$$
S \to AB, \ A \to CD, \ B \to CE, \, E \to DF, \ F \to GH, \ G \to JJ
$$
(the terminal symbols are $C,D,H,J$). The triple set $M$ is
$$
M = \{ (C,1,F), (D,2,F) \}.
$$
A valid sequence is for instance $(S,\ell,A) (A, r, D) (D, 2, F) (F,\ell, J)$.
It represents the circled $g$-labelled node in Figure~\ref{fig-tree-ex}.
\end{example}
From a valid sequence $\gamma$ we can clearly compute in time $O(1)$ the label of the tree node 
represented by $\gamma$. Let us denote this terminal symbol with $\mathsf{label}(\gamma)$:
If $\gamma$ ends with the triple $(A,e,B)$, then we have $B \in N_2$
and  $\mathsf{label}(\gamma)$ is the unique terminal symbol in $\rhs(B)$. If $\gamma = \varepsilon$, then 
$S \in N_c$ and $\mathsf{label}(\gamma)$ is the unique terminal symbol in $\rhs(S)$.

Using valid sequences of $\dA$ it is easy to do a single navigation step in constant time.
Let us fix a valid sequence $\gamma$.
We consider the following possible navigation steps:
Move to the parent node  (if it exists) and move to the $i$th child (if it exists).
Consider for instance the navigation to the $i$th child.
First we check whether $1 \leq i \leq r$, where $r$ is the rank of $\mathsf{label}(\gamma)$.
If this does not hold, then we return undefined. Otherwise, $\gamma$ does not represent a 
leaf of $\val(\dA)$ and therefore $\gamma$ can be neither empty nor
end with a triple from $M \cup (N \times \{\ell,r\} \times N_c)$. Let 
$\beta \neq \varepsilon$ be the maximal suffix of $\gamma$, which belongs to $(L \cup R)^*$.
Then $\beta$ represents a path in the derivation tree of the string SLP $\dB$ that is
rooted in a certain nonterminal $A \in N_a$ and that leads to a certain nonterminal $B \in N_d$.
This path corresponds to a node of $\val_{\dA}(A)$ that is located on the spine path of $A$.
We can now apply our SLP-navigation algorithms $\mathsf{left}$ and $\mathsf{right}$ to the sequence $\beta$
in order to move up or down on the spine path.  
More precisely, let $\beta$ end with the triple $(C,d,B)$ (we must have $d \in \{\ell,r\}$ and $B \in N_d$).  We can now distinguish the following
cases, where $f(A_1, \ldots, A_{j-1}, x, A_{j+1},\ldots, A_n)$ is the right-hand side of $B$:
\begin{enumerate}[(i)]
\item $i \neq j$: Then we obtain the $i$th child by appending to $\gamma$
the triple $(B,i,A_i) \in M$ followed by the path that represents the root of $A_i$ (which consists of at most one triple).
\item $i=j$: Then we obtain the $i$th
child of the current node by moving down on the spine path. Thus, we replace the suffix $\beta$ by
$\mathsf{right}(\beta)$. 
\end{enumerate}
Similar argument can be used to navigate to the parent node.

\begin{algorithm}[t]
\SetKwComment{Comment}{(}{)}
\KwData{valid sequence $\gamma$}
\If{$\gamma \in \{\varepsilon\} \cup L$}
  {\Return undefined}
\If{$\gamma$ belongs to $\alpha \cdot M \cdot L$ or $\alpha \cdot M$ for some $\alpha \in (L \cup R \cup M)^*$}
  {\Return $\alpha$ \Comment*[r]{Note that $\alpha$ must be a valid sequence.}}
let $\gamma = \alpha \cdot \beta$, where $\alpha \in  \{\varepsilon\} \cup (L \cup R \cup M)^* \cdot M$ and $\beta \in (L \cup R)^+$ \;
\Return $\alpha \cdot$ left$(\beta)$  \Comment*[r]{We have $\beta \not\in L$, hence left$(\beta) \neq$ undefined.}
\caption{parent$(\gamma)$  \label{parent}}
\end{algorithm}
\begin{algorithm}[btb]
\SetKwComment{Comment}{(}{)}
\KwData{valid sequence $\gamma$, number $i$}
\If{$i > r$, where $r$ is the rank of $\mathsf{label}(\gamma)$}
  {\Return undefined}
let $\gamma = \alpha \cdot \beta$, where   $\alpha \in  \{\varepsilon\} \cup (L \cup R \cup M)^* \cdot M$ and $\beta \in (L \cup R)^+$ \;
{let $\beta$ end with the triple $(C,d,B)$\Comment*[r]{We must have $B \in N_d$.}}
let $\rhs(B) = f(A_1, \ldots, A_{j-1},x_1,A_{j+1},\ldots, A_n)$ \;
\If{$i = j$}{\Return $\alpha \cdot \text{right}(\beta)$}
\Return $\alpha \cdot \beta \cdot (B, i, A_i) \cdot  \text{root}(A_i)$
\caption{child$(\gamma,i)$  \label{i-th-child}}
\end{algorithm}

Algorithm~\ref{parent} shows the pseudo code for moving to the parent node,
and Algorithm~\ref{i-th-child} shows the pseudo code for moving to the
$i$-th child. If this node does not exist, then  ``undefined'' is returned.
To make the code more readable we denote the concatenation operator for sequences of triples (or sets of triple sequences)
with ``$\cdot$''. We make use of the procedures $\mathsf{left}$
and $\mathsf{right}$ from Section~\ref{sec-string-trav}, applied
to the SLP $\dB$ derived from our TSLP $\dA$.  Note that in Algorithm~\ref{parent} we apply in the final
case the procedure
$\mathsf{left}$ to the maximal suffix $\beta$ from $(L \cup R)^+$
of the current valid sequence $\gamma$
(and similarly for Algorithm~\ref{i-th-child}). To provide an $O(1)$ time implementation
we do not copy the sequence $\beta$ and pass it to $\mathsf{left}$ (which is not possible 
in constant time) but apply $\mathsf{left}$
directly to $\gamma$.
The right-most triple from $M$ in $\gamma$ (if it exists) works as a left-end marker.
Algorithm~\ref{i-th-child} uses the procedure $\mathsf{root}(A)$ (with $A$ of rank $0$).
This procedure is not shown explicitly: It simply returns $\epsilon$ if $A\in N_2$,
and otherwise returns $(A,\ell,\omega_L(A))$ (recall $\omega_L$ from Page~\pageref{page-omega-L}).
Hence, it returns the representation of the root node of $\val_{\dA}(A)$.

The following theorem summarizes the results of this section:

\begin{theorem}
Given a monadic TSLP $\dA$ we can compute in linear time on a word RAM with register length $O(\log |\dA|)$ a data structure of size $O(|\dA|)$ that
allows to do the following computations in time $O(1)$, where $\gamma$ is a valid sequence that represents the tree node $v$: (1)~
Compute the valid sequence for the parent node of $v$, and (2) compute the valid sequence for the $i$th child of $v$.
\end{theorem}

 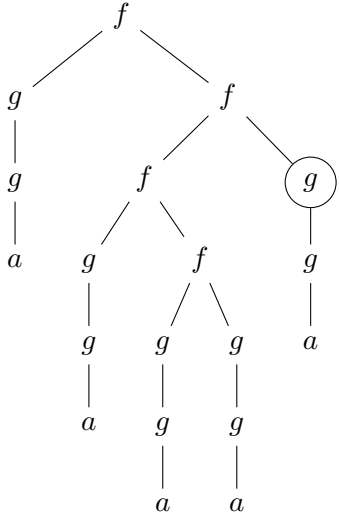
\begin{figure}[t]
\begin{center}
\begin{forest}
[ $f$,
  [ $g$ [ $g$ [ $a$ ]]]
  [ $f$
    [ $f$
      [ $g$ [ $g$ [ $a$ ]]]
      [ $f$
        [ $g$ [ $g$ [ $a$ ]]]
        [ $g$ [ $g$ [ $a$ ]]]
      ]
    ]
    [ $g$,circle,draw [$g$ [$a$ ]]]
  ]
]
\end{forest}
\end{center}
\caption{\label{fig-tree-ex} The tree produced from the TSLP in
Example~\ref{ex-tree-traversal}}
\end{figure}

 \section{Equality checks} \label{sec-equality}

Consider a monadic TSLP $\mc G = (N,\rhs, S)$, where again for every nonterminal $A \in N$, $\rhs(A)$ has one of the following four
forms:
\begin{enumerate}[(a)]
\item $B(C)$ for $B, C \in N$ (and $A$ has rank $0$)
\item $B(C(x))$ for $B, C \in N$ (and $A$ has rank $1$)
\item $a \in \mc F_0$ (and $A$ has rank $0$) 
\item $f(A_1, \ldots, A_{i-1},x,A_{i+1},\ldots, A_n)$ for $A_1, \ldots, A_{i-1}, A_{i+1}, \ldots, A_n \in N$,
$f \in \mc F_n$, $n \geq 1$
 (and $A$ has rank $1$)
\end{enumerate}
Let $t = \val(\dA)$ be the tree produced by $\mc G$.
The goal of this section is to extend the navigation algorithm from the previous section such that
for two nodes of $t$ (represented by valid sequences) we can test in $O(1)$ time whether the subtrees
rooted at the two nodes are equal.
Recall from Section~\ref{sec-comp-model} that we use the word RAM model with registers of length $O(\log |t|)$.
This is the same assumption that is also used in \cite{BilleGLW13,BLRSSW15}.
As before, the preprocessed data structure will have size $O(|\mc G|)$ and the query time will be $O(1)$.
But this time, the preprocessing time will be polynomial in the TSLP size
$|\mc G|$ and not just linear. It will be hard to reduce this to linear time preprocessing. The best known algorithm
for checking equality of SLP-compressed strings has quadratic complexity \cite{Jez15}, and from two SLPs $\dB_1$
and $\dB_2$ we can easily compute a TSLP for a tree $t$ whose root has two children in which $\val(\dB_1)$
and $\val(\dB_2)$ are rooted as linear chains.

We assume that $\mc G$ is {\em reduced} in the sense that $\val_{\dA}(A) \neq \val_{\dA}(B)$ for all
$A,B \in N$ with $A \neq B$.
Reducing the TSLP does not increase its size
and the reduction can be done in polynomial time (recall that we allow polynomial time preprocessing) by Proposition~\ref{prop-TSLP-eq}.
To motivate the forthcoming definitions, we first
give an example of two equal subtrees produced by a single TSLP.
\begin{example}  \label{ex-TSLP-s(A)}
Consider the TSLP $\dA$ with the following rules:
\begin{align*}
& S \to G(A), \ A \to a, \ B(x) \to f(A,x), \ C \to B(A), \ D(x) \to f(C,x), \\
& E \to D(C), \ F(x) \to f(x,E), \ G(x) \to H(I(x)), \ H(x) \to F(B(x)), \\
& I(x) \to D(B(x))
\end{align*}
The reader can check that this TSLP is reduced. The caterpillar tree of $S$ is
given in Figure~\ref{fig-tree-caterpillar}. The tree produced by the subtree
surrounded by a square is the same as the one produced by $E$.
\end{example}
 \begin{figure}[t]
\begin{center}
\begin{forest}
[ $f$
  [ $f$
    [ $A$ ]
    [ $f$,tikz={\node [draw,fit to tree]{};}
      [ $C$ ]
      [ $f$
        [ $A$ ]
        [ $a$ ]
      ]
    ]
  ]
  [ $E$ ]
]
\end{forest}
\end{center}
\caption{\label{fig-tree-caterpillar}
The caterpillar tree produced from $S$ of the TSLP in
Example~\ref{ex-TSLP-s(A)}}
\end{figure}
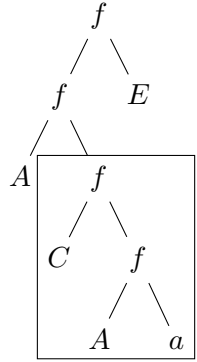
We use the notations introduced in the previous section.
For a string $$w = A_1 A_2 \cdots A_n A_{n+1} \in N_d^* N_c$$ we define
$$
\val_{\dA}(w) = \val_{\dA}(A_1(A_2(\cdots A_n(A_{n+1}) \cdots))).
$$
Recall the definition of the SLP $\dB = (N_1, N_2, \rhs_1)$ from Section~\ref{sec-TSLP-trav}:
If $A \in N_1$ with $\rhs(A) = B(C)$ or  $\rhs(A) = B(C(x))$, then $\rhs_1(A) = BC$.
So, for every $A \in N_a$
we have $\val_{\dB}(A) = A_1 A_2 \cdots A_n A_{n+1}$ for some $n \geq 1$, where
$A_i \in N_d$ for $1 \leq i \leq n$ and $A_{n+1} \in N_c$.
Let $\ell(A) = n+1$ (this is the length of the spine path of $A$), $A[i:j] = A_i A_{i+1} \cdots A_j$,
$A[i:] = A_i \cdots A_n A_{n+1}$, and $A[:i] = A_1  \cdots A_i$.
We define $s(A)$ as the smallest number $i \geq 2$ such that
$\val_{\dA}(A[i:]) = \val_{\dA}(B)$ for some
nonterminal $B \in N$ of rank zero. This unique nonterminal $B$ is denoted
with $A'$. Moreover, let $r_A(x) = \rhs(A_{s(A)-1})$ be the right-hand
side of $A_{s(A)-1} \in N_d$. Hence, $r_A(x)$ is a tree of the form
$f(X_1, \ldots, X_{i-1},x,X_{i+1},\ldots, X_m)$ for $X_1, \ldots, X_{i-1}, X_{i+1}, \ldots, X_m \in N$, $f \in \mc F_m$, $m \geq 1$.
With these notations, we have $\val_{\dA}(A[s(A)-1:]) = \val_{\dA}(r_A(A'))$.
Note that $s(A)$, $r_A$, and $A'$ are well-defined since
$\val_{\dA}(A[n+1:]) = \val_{\dA}(A_{n+1})$ and $A_{n+1}$ has rank zero.

\addtocounter{example}{-1}
\begin{example}[Continued]
The corresponding SLP $\dB$ has the following rules:
$$
S \to GA, \ C \to BA, \ E \to DC, \ G \to HI, \ H \to FB, \ I \to DB
$$
We have $\val_{\dB}(S) = FBDBA$. Moreover,
$\val_{\dA}(DBA) = f(f(a,a),f(a,a)) = \val_{\dA}(E)$, but
$\val_{\dA}(BDBA) = f(a,f(f(a,a),f(a,a)))$ is not equal to one of the trees $\val_{\dA}(X)$
for a nonterminal $X$ of rank zero. Hence, we have $S' = E$, $s(S) = 3$, and $r_S = \rhs(B) = f(A,x)$.
\end{example}

\begin{lemma} \label{lemma-precomp}
For every nonterminal $A \in N_a$, we can compute $s(A), r_A$ and $A'$ in polynomial time.
\end{lemma}

\begin{proof}
We first compute the position $s(A)$ as follows. Let $B_1, \ldots, B_k$ be a list of all nonterminals of
rank zero. We can compute the size $n_i = |\val_{\dA}(B_i)|$ by a simple bottom-up computation.
Let us assume w.l.o.g.~that $n_1 \leq n_2 \leq  \cdots \leq n_k$.
Note that there can be only one position $1 \leq s_i \leq \ell(A)$ ($1 \leq i \leq k$) on the spine path of $A$ such that
the tree $\val_{\dA}(A[s_i:])$ has size $n_i$. This position (if it exists) can be found in polynomial time using binary search.
Note that for a given position $1 \leq s \leq \ell(A)$ we can compute the size of tree $\val_{\dA}(A[s:])$ in polynomial time
by first computing a TSLP for this tree (using the second statement in Proposition~\ref{prop-SLP-alg}) and then compute the size of the produced tree bottom-up.
Once the positions $s_1, \ldots, s_k$ are computed, we can compute in polynomial time TSLPs for the trees
$t_i = \val_{\dA}(A[s_i:])$. The position $s(A)$ is the smallest $s_i \geq 2$ such that $t_i$ is equal to one of the trees
$\val_{\dA}(B_j)$, and the latter can be checked in polynomial time by Proposition~\ref{prop-TSLP-eq}.
This also yields the nonterminal $A'$. Finally, the right-hand side $r_A$ can be computed
by computing in polynomial time the symbol $B \in N_2$ in $\val_{\dB}(A)$ at position $s(A)-1$ by the first statement of
Proposition~\ref{prop-SLP-alg}. Then, $r_A = \rhs(B)$.
\end{proof}

\begin{lemma} \label{lemma-eq}
For all $A, B \in N_a$ and all $1 \leq i < s(A)$, $1 \leq j < s(B)$, the following two conditions are equivalent:
\begin{enumerate}[(i)]
\item $\val_{\dA}(A[i:]) = \val_{\dA}(B[j:])$
\item $\val_{\dB}(A[i:s(A)-2]) = \val_{\dB}(B[j:s(B)-2])$ and $r_A(A') = r_B(B')$.
\end{enumerate}
\end{lemma}
\begin{proof}
Clearly, (ii) implies (i). Now assume that $\val_{\dA}(A[i:]) = \val_{\dA}(B[j:])$ holds. Let
$\val_{\dB}(A) = A_1 A_2 \cdots A_n A_{n+1}$ and $\val_{\dB}(B) = B_1 B_2 \cdots B_m B_{m+1}$.
By induction on $i$ and $j$, we show that $\val_{\dB}(A[i:s(A)-2]) = \val_{\dB}(B[j:s(B)-2])$ and $r_A(A') = r_B(B')$.

\medskip
\noindent
{\em Case 1.} $i = s(A)-1$. By the definition of $s(A)$, this implies that the tree $\val_{\dA}(A[i:])$
has the form $f(\val_{\dA}(X_1),\ldots, \val_{\dA}(X_k))$, where $f \in \mc F_k$, $k \geq 1$, and $X_1, \ldots, X_k$ are
nonterminals of rank zero. Hence, $$\val_{\dA}(B[j:]) = f(\val_{\dA}(X_1),\ldots, \val_{\dA}(X_k)).$$
But this implies that $\val_{\dA}(B[j+1:])$ is equal to one of the trees $\val_{\dA}(X_l)$. We therefore
get $j = s(B)-1$. Thus, $\val_{\dB}(A[i:s(A)-2]) = \varepsilon = \val_{\dB}(B[j:s(B)-2])$. Moreover,
$\val_{\dA}(r_A(A')) = \val_{\dA}(A[s(A)-1:]) = \val_{\dA}(B[s(B)-1:]) = \val_{\dA}(r_B(B'))$. Since $\dA$ is reduced,
we get $r_A(A') = r_B(B')$.

\medskip
\noindent
{\em Case 2.} $j = s(B)-1$. This case is symmetric to Case~1.

\medskip
\noindent
{\em Case 3.} $i < s(A)-1$ and $j < s(B)-1$. We claim that $A_i = B_j$. If $A_i \neq B_j$, then, since $\dA$ is reduced,
we have $\rhs(A_i) \neq \rhs(B_j)$. With $\val_{\dA}(A[i:]) = \val_{\dA}(B[j:])$, this implies that each of the trees
$\val_{\dA}(A[i+1:])$ and $\val_{\dA}(B[j+1:])$ is equal to a tree of the form $\val_{\dA}(X)$ for $X$ a nonterminal of rank zero.
But this contradicts $i+1 < s(A)$ as well as $j+1 < s(B)$. Hence, we have $A_i = B_j$. Moreover, $\rhs(A_i) = \rhs(B_j)$ implies
with $\val_{\dA}(A[i:]) = \val_{\dA}(B[j:])$ that $\val_{\dA}(A[i+1:]) = \val_{\dA}(B[j+1:])$. We can now conclude with induction.
\end{proof}
Consider a valid sequence
$\gamma \in (L \cup R \cup M)^*$ for $\dA$. We can uniquely factorize $\gamma$
as
\begin{equation} \label{eq-factorized-sequence}
\gamma = \gamma_1 (A_1,k_1,B_1) \gamma_2 (A_2,k_2,B_2) \cdots \gamma_{n-1} (A_{n-1},k_{n-1},B_{n-1}) \gamma_n,
\end{equation}
where $\gamma_i \in (L \cup R)^*$ ($1 \leq i \leq n$) and $(A_i,k_i,B_i) \in M$ ($1 \leq i \leq n-1$).
To simplify the notation, let us set $B_0 = S$.
Hence,  every $\gamma_i$ is either empty or a  valid $B_{i-1}$-sequence for the SLP $\dB$,
and we have defined the position $\pos(\gamma_i)$ in the string $\val_{\dB}(B_{i-1})$ according to Section~\ref{sec-string-trav}.
It is easy to modify our traversal algorithms from the previous section such that for every $1 \leq i \leq n$
we store in the sequence $\gamma$ also the nonterminal $B_{i-1}$ and the number $\pos(\gamma_i)$ right after $\gamma_i$ (if $\gamma_i \neq \varepsilon$), i.e.,
just before $(A_i, k_i, B_i)$. The number $\pos(\gamma_i)$ has to be incremented (resp., decremented) each time one 
moves down (resp., up) in the spine path for $B_{i-1}$.
We will not explicitly write these nonterminals and positions in valid sequences in order do not
complicate the notation.

We would like to use Lemma~\ref{lemma-eq} for equality checks. To do this, we have to assume that
$\pos(\gamma_i) < s(B_{i-1})$ in \eqref{eq-factorized-sequence} for every $1 \leq i \leq n$ with $\gamma_i \neq \varepsilon$
(this corresponds to the assumptions $1 \leq i < s(A)$ and $1 \leq j < s(B)$ in Lemma~\ref{lemma-eq}).
To do this, we have to modify our traversal algorithms from the previous section as follows: Assume that the current
valid sequence is $\gamma$ from  \eqref{eq-factorized-sequence}. As remarked above, we store the numbers
$\pos(\gamma_i)$ right after each $\gamma_i$. Assume that the final number $\pos(\gamma_n)$ has reached the value
$s(B_{n-1})-1$ and we want to move to the $i$th child of the current node. 
We proceed as in Algorithm~\ref{i-th-child}
with one exception: In case (ii) (Section~\ref{sec-TSLP-trav}) we would increase $\pos(\gamma_i)$ to $s(B_{n-1})$.
To avoid this, we start a new valid sequence for the root of the tree $\val_{\dA}(B'_{n-1})$. Note that by the definition of $B'_{n-1}$,
this is exactly the tree rooted at the $i$th child of the node represented by $\gamma$.
So, we can continue the traversal in the tree $\val_{\dA}(B'_{n-1})$. Therefore,
we continue with the sequence $\gamma\mid\mathsf{root}(B'_{n-1})$, where $\mid$ is a separator symbol, and the $\mathsf{root}$-function is
defined at the end of Section~\ref{sec-TSLP-trav}. The navigation to the parent node can be easily adapted as well.
The only new case that we have to add to Algorithm~\ref{parent} is for $\gamma = \alpha\mid\beta$, where $\beta \in (L \cup R)^+$.
In that case, we compute $\beta' = \mathsf{left}(\beta)$ and return $\alpha\mid\beta'$ if $\beta'$ is not undefined, and
$\alpha$ otherwise. Thus, the separator symbol $\mid$ is treated in the same way as triples from $M$.

Let us now consider two sequences $\gamma_1$ and $\gamma_2$ (that may contain the separator symbol $\mid$ as explained in the previous paragraph).
Let $v_i$ be the node of $\val(\dA)$ represented by $\gamma_i$ and let $t_i$ be the subtree of $\val(\dA)$ rooted in $v_i$.
We want to check in time $O(1)$ whether $t_1 = t_2$.
We can first compute in time $O(1)$ the labels $\mathsf{label}(\gamma_1)$ and $\mathsf{label}(\gamma_2)$
of the nodes $v_1$ and $v_2$, respectively. In case one of these labels belongs to $\mc F_0$ (i.e., one of the nodes
$v_1$, $v_2$ is a leaf) we can easily determine whether $t_1 = t_2$. Hence, we can assume that 
neither $v_1$ nor $v_2$ is a leaf. In particular we can assume that 
$\gamma_1 \neq \varepsilon \neq \gamma_2$ (recall that $\varepsilon$ is a valid sequence only in case $\val(\dA)$ consists
of a single node) and that neither $\gamma_1$ nor $\gamma_2$ ends with a triple from $M$. 
Let us factorize $\gamma_i$ as $\gamma_i = \alpha_i \beta_i$,
where $\beta_i$ is the maximal suffix of $\gamma_i$ that
belongs to $(L \cup R)^*$. Hence, we have $\beta_1  \neq \varepsilon \neq \beta_2$.

Assume that $\beta_i$ is a valid $C_i$-sequence of $\dB$,
and that $n_i = \pos(\beta_i)$. Thus, the suffix $\beta_i$ represents the $n_i$-th leaf of the derivation tree of $\dB$ with root $C_i$.
Recall that we store $C_i$ and $n_i$ at the very end of our sequence $\gamma_i$. Hence, we have constant time access to $C_i$ and $n_i$.
We have $C_i \in N_a$ and $n_i < s(C_i)$. With the notation introduced before, we get $t_i = \val_{\dA}(C_i[n_i:])$. Since
$n_1 < s(C_1)$ and $n_2 < s(C_2)$, Lemma~\ref{lemma-eq} implies that $t_1 = t_2$ if and only
if the following two conditions hold:
\begin{enumerate}[(i)]
\item $\val_{\dB}(C_1[n_1:s(C_1)-2]) = \val_{\dB}(C_2[n_2:s(C_2)-2])$ and
\item $r_{C_1}(C_1') = r_{C_2}(C_2')$.
\end{enumerate}
Condition (ii) can be checked in time $O(1)$, since we can precompute in polynomial time $r_A$ and $A'$ for every
$A \in N_a$ by Lemma~\ref{lemma-precomp}.
So, let us concentrate on condition (i). First, we check whether
$s(C_1)-n_1 = s(C_2)-n_2$. If not, then the lengths of
$\val_{\dB}(C_1[n_1:s(C_1)-2])$ and $\val_{\dB}(C_2[n_2:s(C_2)-2])$ differ and we cannot have equality.
Hence, assume that $k := s(C_1)-1-n_1 = s(C_2)-1-n_2$.
Let $\ell$ be the length of the longest common suffix of $\val_{\dB}(C_1[:s(C_1)-2])$ and $\val_{\dB}(C_2[:s(C_2)-2])$.
Then, it remains to check whether $k \leq \ell$. Clearly, in space $O(|\dA|)$ we cannot store explicitly all these lengths $\ell$
for all $C_1, C_2 \in N_a$. Instead, we precompute in polynomial time a modified Patricia tree for the set of strings
$w_A := \val_{\dB}(A[:s(A)-2])^{\text{rev}} \$$ ($\$$ is a new symbol that is appended in order
make the set of strings prefix-free and $w^{\text{rev}}$ is the string $w$ reversed)
for $A \in N_a$. Then, we need to compute in $O(1)$ time
the length of the longest common prefix for two of these strings $w_A$ and $w_B$.
Recall that the  Patricia tree for a set of strings $w_1, \ldots, w_n$ is obtained from the trie for the prefixes of the $w_i$
by eliminating nodes with a single child. But instead of labeling edges of the Patricia tree with factors of the $w_i$, we
label every internal node with the length of the strings that leads from the root to the node.
Let us give an example instead of a formal definition:

\begin{example} \label{ex-patricia}
Consider the strings $w_A = abba\$$, $w_B = abbb\$$, $w_C = ba\$$, $w_D =
baba\$$ and $w_E =babb\$$.
Figure~\ref{fig-patricia} shows their Patricia tree (left) and the
modified Patricia tree (right).
\end{example}
\begin{figure}[t]
\begin{center}
\begin{forest}
[ $\bullet$,
  [ $\bullet$, edge label={node[midway,above,sloped]{$bba$}}
    [ $A$, edge label={node[midway,above,sloped]{$\$a$}} ]
    [ $B$, edge label={node[midway,above,sloped]{$b\$$}} ]
  ]
  [ $\bullet$, edge label={node[midway,above,sloped]{$ba$}}
    [ $C$, edge label={node[midway,above,sloped]{$\$$}} ]
    [ $\bullet$, edge label={node[midway,above,sloped]{$b$}}
      [ $D$, edge label={node[midway,above,sloped]{$\$a$}} ]
      [ $E$, edge label={node[midway,above,sloped]{$b\$$}} ]
    ]
  ]
]
\end{forest}
\qquad
\begin{forest}
[ $0$,
  [ $3$
    [ $A$ ]
    [ $B$ ]
  ]
  [ $2$
    [ $C$ ]
    [ $3$
      [ $D$ ]
      [ $E$ ]
    ]
  ]
]
\end{forest}
\end{center}
\caption{\label{fig-patricia} The Patricia tree (left) and the modified
Patricia tree (right) for Example~\ref{ex-patricia}}
\end{figure}
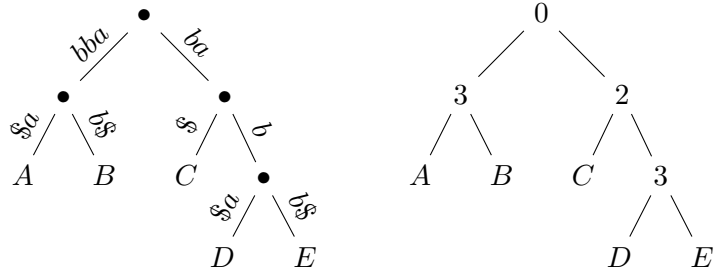
Since our modified Patricia tree has $|N_a|$ many leaves (one for each $A \in N_a$)
and every internal node has at least two children, we have at most $2|N_a|-1$ many nodes
in the tree and every internal node is labeled with an $(\log |t|)$-bit number (note that
the length of every string $\val_{\dB}(A)$ ($A \in N_a$) is bounded by $|t|$.
Hence, on the word RAM model, we can store the modified Patricia tree in space
$O(|\dA|)$. Finally, the length of the longest common prefix of two string $w_A$ and $w_B$
can be obtained by computing the lowest common ancestor of the two leaves corresponding
to the strings $w_A$ and $w_B$ in the Patricia tree. The number stored in
the lowest common ancestor is the  length of the longest common prefix of $w_A$ and $w_B$.
Using a data structure for computing lowest common ancestors in time $O(1)$
\cite{BenderF00,DBLP:journals/siamcomp/SchieberV88}, we obtain an $O(1)$-time implementation
of subtree equality checking. Finally, from Proposition~\ref{prop-SLP-alg} it follows that
the modified Patricia tree for the strings $w_A$ ($A \in N_a$)
can be precomputed in polynomial time.

The following theorem is the main result of this section:
\begin{theorem} \label{thm-traverse+equality}
Given a monadic TSLP $\dA$ for a tree $t = \val(\dA)$
we can compute in polynomial time on a word RAM with register length $O(\log |t|)$ a data structure of size $O(|\dA|)$ that
allows to do the following computations in time $O(1)$, where $\gamma$ and $\gamma'$ are valid sequences (as modified in this section) that represent the tree nodes $v$
and $v'$, respectively: (i)~
Compute the valid sequence for the parent node of $v$, (ii) compute the valid sequence for the $i$th child of $v$, and (iii) check whether the subtrees
rooted in $v$ and $v'$ are equal.
\end{theorem}

\section{Discussion} \label{conc}

We have presented a data structure to traverse grammar-compressed ranked trees with constant delay and to check
equality of subtrees.
The solution is based on the ideas of~\cite{GasieniecKPS05} and the fact that
next link queries can be answered in time $O(1)$ (after linear time preprocessing).
It would be interesting to develop an efficient implementation of the technique.
Next link queries can be implemented in many different ways.
One solution given in~\cite{GasieniecKPS05} is based on a variant of the lowest common ancestor algorithm
due to Schieber and Vishkin~\cite{DBLP:journals/siamcomp/SchieberV88}
(described in~\cite{DBLP:books/cu/Gusfield1997}).
Another solution is to use \emph{level-ancestor queries} (together with depth-queries),
as are available in implementation of succinct tree data structures
(e.g., the one of Navarro and Sadakane~\cite{DBLP:journals/talg/NavarroS14}).
Another alternative is to store the first-child/next-sibling encoded binary tree of the original
tries. The first-child/next-sibling encoding is defined for ordered trees. In our situation,
we have to answer next-link queries for the tries $T_L(a)$ and $T_R(a)$ for $a \in \Sigma$,
which are unordered. Hence we order the children of a node in an arbitrary way.
Then the next link of $v_1$ above $v_2$ is equal  to the \emph{lowest common ancestor}
of $v_2$ and the last child of $v_1$ in the original tree. This observation allows to use simple
and efficient lowest common ancestor data structures like the one of Bender and Farach-Colton~\cite{BenderF00}.

Recall that we used polynomial time preprocessing to built up our data structure for subtree equality checks.
It remains to come up with a more precise time bound. We already argued that the problem is at least as difficult
as checking equality of SLP-compressed strings, for which the best known algorithm is quadratic. It would be interesting,
to show that our preprocessing has the same time complexity as checking equality of SLP-compressed strings.




\def\cprime{$'$} \def\cprime{$'$}

\end{document}